\documentclass[review]{elsarticle}

\usepackage{hyperref}
\usepackage{amssymb}
\usepackage{amsmath}
\usepackage{amsthm}
\usepackage{enumitem}
\usepackage{tikz}

\usetikzlibrary{decorations.pathreplacing}

\newtheorem{theorem}{Theorem}
\newtheorem{lemma}{Lemma}
\newtheorem{definition}{Definition}

\newenvironment{case}[2]{\noindent \underline{Case #1}: #2 \\}{}

\newcommand{\nat}{\mathbb{N}}
\newcommand{\timeSet}{T}
\newcommand{\nodeSet}{V}
\newcommand{\epochSet}{E}
\newcommand{\boxSet}{B}
\newcommand{\initialSet}{X}
\newcommand{\configSet}{C}
\newcommand{\state}{S}
\newcommand{\aFun}{F}

\newcommand{\async}{\delta}
\newcommand{\accordant}{accordant}
\newcommand{\UD}{\"Uresin \& Dubois}

\newcommand{\accSet}{A^p}

\begin{document}

\begin{frontmatter}

\title{Dynamic asynchronous iterations}

\author{Matthew L. Daggitt*}
\author{Timothy G. Griffin}
\address{Department of Computer Science and Technology, University of Cambridge, UK}

\begin{abstract}
Many problems can be solved by iteration  by multiple participants (processors, servers, routers etc.). Previous mathematical models for such asynchronous iterations assume a single function being iterated by a fixed set of participants. We will call such iterations \emph{static} since the system's configuration does not change. However in several real-world examples, such as inter-domain routing, both the function being iterated and the set of participants change frequently while the system continues to
function. In this paper we extend \UD{}'s work on static iterations to develop a model for this class of \emph{dynamic} or \emph{always on} asynchronous iterations. We explore what it means for such an iteration to be implemented correctly, and then prove two different conditions on the set of iterated functions that guarantee the full asynchronous iteration satisfies this new definition of correctness. These results have been formalised in Agda and the resulting library is publicly available.
\end{abstract}


\end{frontmatter}

\section{Introduction}

Let $\state$ be a set. Iterative algorithms aim to find a fixed point $x^*$ for some function ${\aFun : \state \rightarrow \state}$ by starting from an initial state $x \in \state$ and calculating the sequence:
\begin{equation*}
x,\ \aFun(x),\ \aFun^2(x), \aFun^3(x),\ ...
\end{equation*}
If a number of iterations $k^*$ is found such that $\aFun^{k^*}(x) = \aFun^{k^*+1}(x)$ then $\aFun^{k^*}(x)$ is a fixed point $x^*$. Whether or not there exists such a $k^*$ depends on both the properties of the iterated function $\aFun$ and the initial state chosen $x$. It should be noted that this paper is only interested in functions $\aFun$ which converge to a unique fixed point, i.e. the same $x^*$ is reached no matter which initial state $x$ the iteration starts from.

In a distributed version of the iteration, both the set $\state$ and the function $\aFun$ are assumed to be decomposable into $n$ parts:
\begin{align*}
\state &= \state_1 \times \state_2 \times ... \times \state_n \\
\aFun  &= (\aFun_1 , \aFun_2, ..., \aFun_n)
\end{align*}
where $\aFun_i : \state \rightarrow \state_i$ computes the $i^{th}$ component of the new state. Each node $i$ repeatedly iterates $\aFun_i$ on its local view of the current state of the iteration, and propagates its stream of updated values to other nodes so that they may incorporate them in their own iteration. In an \emph{asynchronous} distributed iteration, the timings between nodes are not actively synchronised. A formal model, $\async$, for such an asynchronous iteration is described in Section~\ref{sec:static-model}.

Frommer~\&~Syzld~\cite{frommer00asynchronous} provide a survey of the literature describing when such asynchronous iterations are guaranteed to converge to a unique fixed point. One of the unifying features of these results is that they only require conditions on the function $\aFun$ and hence users may prove an asynchronous iteration converges without ever directly reasoning about  unreliable communication or asynchronous event orderings. Applications of these results include routing~\cite{chau06routing,daggitt18convergence,duco03pathalgebra}, programming language design~\cite{edwards03semantics}, peer-to-peer protocols~\cite{ko08ptp}, and numerical simulations~\cite{chau05numerical}. Other recent applications of asynchronous iterations include \cite{wolfson2019modeling, magoules2021asynchronous, magoules2018asynchronous}, while \cite{spiteri2020parallel} and~\cite{bahi2007parallel} provide useful surveys of asynchronous iterations in general.

However there are two main drawbacks to the models used in the current literature. Firstly, they assume the set of participating nodes remains constant over time. While this may be reasonable when modelling an iterative process run on a multi-core computer, it is unrealistic when reasoning about truly distributed ``always on'' protocols such as routing and consensus algorithms. For example the global Border Gateway Protocol that coordinates routing in the internet has been ``on'' since the early 1990's and has grown from a few dozen routers to millions. During that time the set of participating routers has been completely replaced many times over. The second problem is that the models assume that the function $\aFun$ being iterated remains constant over time. This may not be the case if it depends on some process external to the iteration (e.g. link latencies in routing) or on the set of participants (e.g. resource allocation/consensus/routing algorithms). 

This paper will therefore use the term \emph{static} to refer to the asynchronous iterations previously described in the literature and \emph{dynamic} to refer to this new class of asynchronous iterations in which the set of participating nodes and function being iterated may change over time.

When applying the results in the literature to always-on algorithms, it is common for prior work to either informally argue or implicitly assume that the correctness of a dynamic iteration is an immediate consequence of the correctness of a infinite sequence of suitable static iterations. This line of reasoning is:
\begin{itemize}
\item implicitly argued in Section~4.2 of \cite{chau06routing}.
\item explicitly argued in Section~3.2 of \cite{daggitt18convergence}.
\item implicitly argued in Section~2.4 of \cite{duco03pathalgebra}.
\item discussed and implicitly argued in point (b) of Section 5 in \cite{bertsekas89survey}.
\end{itemize}
The reasoning runs that a dynamic iteration is really a sequence of static iterations, where each new static iteration starts from the final state of the previous static iteration. However this argument is incorrect, as it does not take into account that messages may be shared between the different static iterations in the sequence. For example if node 1 fails, it may still have messages in flight that node 2 will receive in the next static iteration. Not only may this message prevent convergence in the next iteration, the model in the existing literature has no way of even representing messages arriving from nodes that are not participating during the current static iteration.

This paper therefore proposes a new, more general model that can be used to reason about dynamic iterations over both continuous and discrete data. Section~\ref{sec:static-model} of the paper describes one of the most commonly used static models, and discusses some of the surrounding literature. Section~\ref{sec:dynamic-model} then presents our new generalised model for dynamic iterations, and discusses what it means for a dynamic iteration to be ``correct''. Next, Section~\ref{sec:convergent-results} proves two different conditions for a dynamic asynchronous iteration to satisfy this definition of correctness. Importantly, and as with the previous static results of \UD, these conditions only constrain the \emph{synchronous} behaviour of the dynamic system. This means that users of our theorems can prove the correctness of their asynchronous algorithms by purely synchronous reasoning.  Section~\ref{sec:agda} then briefly describes the formalisation of the results in Agda and their application to inter-domain routing protocols. Finally, Section~\ref{sec:conclusion} discusses our concluding thoughts and possible directions for future work.
\section{Static asynchronous iterations}
\label{sec:static-model}

\subsection{Model}

A mathematical model for static asynchronous iterations was standardised by work in the 1970s and 80s~\cite{bertsekas89survey, baudet76asynchronous, uresin86generalized}. The notation and terminology used here is taken from the recent paper~\cite{anonymised} which in turn is based on that used by \UD~\cite{uresin90parallel}.

Assume that the set of times $\timeSet$ is a discrete linear order. Each point in time marks the occurrence of events of interest: for example a node computing an update or a message arriving at a node. The set of times can be represented by~$\mathbb{N}$ but for notational clarity $\timeSet$ will be used. Additionally let $\nodeSet = \{1,2,...,n\}$ be the set of nodes that are participating in the computation.

\begin{definition}[Static schedule]
A static schedule consists of a pair of functions:
\begin{itemize}
\item $\alpha : \timeSet \rightarrow 2^\nodeSet$, the \emph{activation function}, where $\alpha(t)$ is the set of nodes which activate at time~$t$.

\item $\beta : \timeSet \times \nodeSet \times \nodeSet \rightarrow \timeSet$, the \emph{data flow function}, where $\beta(t,i,j)$ is the time at which the latest message node $i$ has received from node $j$ at time $t$ was sent by node $j$.
\end{itemize}
such that:
\begin{enumerate}[label=(SS\arabic*),leftmargin=1.5cm]
\item $\forall i,j,t : \: \beta(t+1,i,j) \leq t$
\label{ass:static-sched-causality}
\end{enumerate}
\end{definition}
The function~$\alpha$ describes when nodes update their values, and the function~$\beta$ tracks how the resulting information moves between nodes. Assumption \ref{ass:static-sched-causality} enforces causality by stating that information may only flow forward in time. Note that this definition does not forbid the data flow function $\beta$ from delaying, losing, reordering or even duplicating messages (see Figure~\ref{fig:schedule}). Prior to recent work~\cite{anonymised}, static schedules were assumed to have two additional assumptions that guaranteed every node continued to activate indefinitely and that every pair of nodes continued to communicate indefinitely. 

\begin{figure}
\centering
\begin{tikzpicture}
\def\a{0}
\def\b{0.7}
\def\c{2}
\def\d{4.5}
\def\e{6}
\def\f{6.5}
\def\g{10}
\def\h{10.5}

\def\braceLabelY{2.7}
\def\braceY{2.25}
\def\brcD{0.05}

\def\th{-0.5}
\def\bh{-1}
\def\ih{2}
\def\jh{0}

\def\tickHeight{0.3}
\def\labelY{-0.03}

\tikzstyle{bracesU}=[thick, decorate, decoration={brace,amplitude=6pt,raise=0pt}]
\tikzstyle{bracesD}=[thick, decorate, decoration={brace,amplitude=6pt,raise=0pt,mirror}]

\draw [thick,dashed,->] (0.7,2) -- (11,2);
\draw [thick,dashed,->] (0.7,0) -- (11,0);
\node at (0,\ih) {Node $i$};
\node at (0,\jh) {Node $j$};

\node at (0,\th) {Time $t$};
\node at (1,\th) {1};
\node at (2,\th) {2};
\node at (3,\th) {3};
\node at (4,\th) {4};
\node at (5,\th) {5};
\node at (6,\th) {6};
\node at (7,\th) {7};
\node at (8,\th) {8};
\node at (9,\th) {9};
\node at (10,\th){10};

\node at (0,\bh) {$\beta(t,i,j)$};
\node at (1,\bh) {0};
\node at (2,\bh) {0};
\node at (3,\bh) {2};
\node at (4,\bh) {1};
\node at (5,\bh) {1};
\node at (6,\bh) {1};
\node at (7,\bh) {1};
\node at (8,\bh) {7};
\node at (9,\bh) {8};
\node at (10,\bh){7};

\draw [->](1,\jh) -- (4,\ih);
\draw [->](2,\jh) -- (3,\ih);

\draw [->](5,\jh) -- (6,1);

\draw [->](7,\jh) -- (8,\ih);
\draw [->](7.5,1) -- (10,\ih);
\draw [->](8,\jh) -- (9,\ih);


\draw[bracesU] (\b+\brcD,\braceY) -- (\d-\brcD,\braceY);
\draw[bracesU] (\d+\brcD,\braceY) -- (\f-\brcD,\braceY);
\draw[bracesU] (\f+\brcD,\braceY) -- (\h-\brcD,\braceY);

\node[align=center] at ({(\b+\d)/2},\braceLabelY){Messages reordered};
\node[align=center] at ({(\d+\f)/2},\braceLabelY){Message lost};
\node[align=center] at ({(\f+\h)/2},\braceLabelY){Message duplicated};
\end{tikzpicture}
\caption{Behaviour of the data flow function
$\beta$. Messages from node $j$ to node $i$ may be reordered, lost or even duplicated. The only constraint is that every message must arrive after it was sent. Reproduced from~\cite{anonymised}}
\label{fig:schedule}
\end{figure}
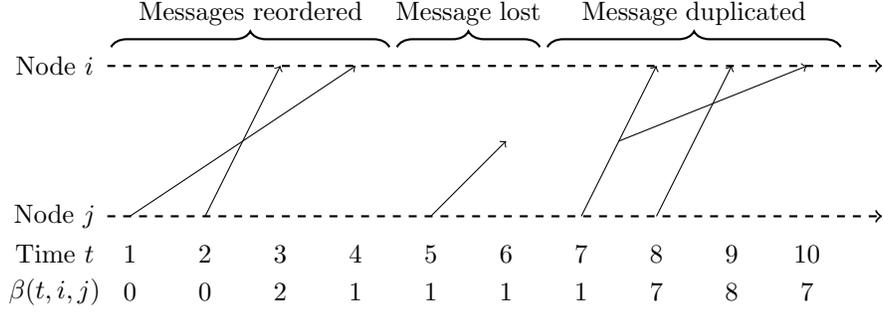

\begin{definition}[Static asynchronous state function]
Given a static schedule $(\alpha, \beta)$ the static asynchronous state function, $\async : \timeSet \rightarrow \state \rightarrow \state$, is defined as follows:
\begin{align*}
\async^t_i(x) &= \begin{cases}
x_{i} & \text{if $t = 0$}\\
\async^{t-1}_i(x) & \text{else if $i \notin \alpha(t)$} \\
\aFun_i(\async^{\beta(t,i,1)}_1(x), \async^{\beta(t,i,2)}_2(x), ... , \async^{\beta(t,i,n)}_n(x)) & \text{otherwise}
\end{cases}
\end{align*}
where $\async^t_i(x)$ is the state of node $i$ at time $t$ when starting from state $x$.
\end{definition}
At time $0$ the iteration is in the initial state $x$. At subsequent times $t$ if node~$i$ is not in the set of active nodes then its state remains unchanged. Otherwise if node $i$ is in the active set of nodes it applies its update function $F_i$ to its \emph{current view} of the global state. For example $\async^{\beta(t,i,1)}_1(x)$ is the state of node $1$ at the time of departure of the most recent message node $i$ has received from node $1$ at time $t$.

\subsection{Correctness}

In order to precisely define when an asynchronous iteration is expected to converge, it is first necessary to discuss what sort of schedules allow an asynchronous iteration to make progress. As mentioned earlier, previous models made the simplifying assumption that every node activates an infinite number of times and every pair of nodes continue to communicate indefinitely. This essentially says that the schedule is well-behaved forever. In contrast \cite{anonymised} built upon the work of \UD{} and their concept of \emph{pseudocycles} and relaxed this condition to only require that schedules must be well-behaved for a finite period of time. This distinction will be important in the dynamic model described later in Section~\ref{sec:dynamic-model}, as a dynamic iteration will only have a finite period of time to converge before either the participants or the function being iterated changes.

\begin{definition}[Static activation period]
A period of time $[t_1,t_2]$ is an activation period for node $i$ if there exists a time $t \in [t_1,t_2]$ such that $i \in \alpha(t)$.
\end{definition}

\begin{definition}[Static expiry period]
A period of time $[t_1,t_2]$ is an expiry period for node $i$ if for all nodes $j$ and times $t \geq t_2$ then $t_1 \leq \beta(t,i,j)$.
\end{definition}
Therefore after an activation period node $i$ is guaranteed to activate at least once. In contrast after an expiry period the node is guaranteed to use only data generated after the start of the expiry period. In other words, all messages in flight to node $i$ at time $t_1$ have either arrived or been lost by time $t_2$.

\begin{definition}[Static pseudocycle]
A period of time $[t_1,t_2]$ is a pseudocycle if for all nodes $i$ there exists a time $t \in [t_1,t_2]$ such that $[t_1,t]$ is an expiry period for node $i$ and $[t,t_2]$ is an activation period for node $i$.
\end{definition}
The term ``pseudocycle'' refers to the fact that during such a period of time the asynchronous iteration will make at least as much progress as that of a single step of the synchronous iteration. This statement will be made formal later on by Lemma~\ref{lem:pseudocycle} in Section~\ref{sec:aco-implies-convergent}. When we informally say that a period of time contains $k$ pseudocycles, we implicitly mean $k$ \textit{disjoint} pseudocycles.

Using the definition of a pseudocycle, it is now possible to define what it means for an asynchronous iteration to converge for schedules that are only well-behaved for a finite amount of time.

\begin{definition}[Static convergence]
The static asynchronous iteration converges over a set of initial states ${\initialSet = \initialSet_1 \times \initialSet_2 \times \ldots \times \initialSet_n} \subseteq \state$ if:
\begin{enumerate}
\item there exists a fixed point $x^*$ for $\aFun$ and a number of iterations $k^*$.
\item for every starting state $x \in X$ and schedule containing at least $k^*$ pseudocycles then there exists a time $t_1$ such that for all $t_2 \geq t_1$ then $\async^{t_2}(x) = x^*$.
\end{enumerate}
\end{definition}

\subsection{Results}
\label{sec:static-results}

The survey paper by Frommer~\&~Syzld~\cite{frommer00asynchronous} provides an overview of the convergence results in the literature for this and other related models. Much of the work has been motivated by iterative algorithms in numerical analysis and consequently many of the proofs of convergence assume that the set $\state$ is equipped with a dense ordering. Unfortunately in fields such as routing, consensus algorithms and others, the set~$\state$ is discrete, and so many of the more common results are inapplicable. However in the late 1980s \UD{}~\cite{uresin90parallel} came up with one of the first conditions for the convergence of discrete asynchronous iterations. Here we use the relaxed version of the conditions as proposed in~\cite{anonymised}.

\begin{definition}[Static ACO]
\label{def:static-aco}
A function $\aFun$ is an asynchronously contracting operator (ACO) if there exists a sequence of sets $\boxSet(k) = \boxSet(k)_1 \times \boxSet(k)_2 \times ... \times \boxSet(k)_n$ for $k \in \mathbb{N}$ such that:
\begin{enumerate}[label=(SA\arabic*),leftmargin=1.5cm]
\item $\forall x \in \state: x \in \boxSet(0) \Rightarrow \aFun(x) \in \boxSet(0)$.
\label{ass:static-aco-closed}
\item $\forall k \in \nat, x \in \state : x \in \boxSet(k) \Rightarrow \aFun(x) \in \boxSet(k+1)$.
\label{ass:static-aco-progress}
\item $\exists k^* \in \nat, x^* \in \state : \forall k \in \nat: k^* \leq k \Rightarrow \boxSet(k) = \{x^*\}$.
\label{ass:static-aco-finish}
\end{enumerate} 
\end{definition}

\begin{theorem}
\label{thm:static-aco-implies-converges}
If function $\aFun$ is an ACO then $\async$ converges deterministically over the set~$B(0)$.
\end{theorem}
\begin{proof}
See \cite{uresin90parallel} \& \cite{anonymised}. 
\end{proof}
\noindent The advantage of the ACO conditions is that they are independent of both $\async$ and the schedule, and so proving that $\async$ converges only requires reasoning about the function $\aFun$.

The conditions require that the state space $\state$ can be divided into a series of nested boxes $\boxSet(k)$ where every application of $\aFun$ moves the state into the next box, and eventually a box $\boxSet(k^*)$ is reached that only contains a single element. See Figure~\ref{fig:aco} for a visualisation. 
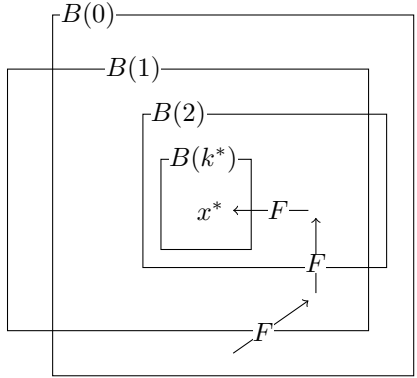
\begin{figure}
\centering
\begin{tikzpicture}

\def\s{1.2}
\def\dx{6}
\def\tx{0.1}

\draw (0*\s,0*\s) 	rectangle (4*\s,4*\s);
\draw (-0.5*\s,0.5*\s) rectangle (3.5*\s,3.4*\s);
\draw (1.0*\s,1.2*\s) rectangle (3.7*\s,2.9*\s);
\draw (1.2*\s,1.4*\s) rectangle (2.2*\s,2.4*\s);

\node[anchor=west,fill=white,inner sep=0pt] at (0*\s+\tx,4*\s) {$\boxSet(0)$};
\node[anchor=west,fill=white,inner sep=0pt] at (0.5*\s+\tx,3.4*\s) {$\boxSet(1)$};
\node[anchor=west,fill=white,inner sep=0pt] at (1.0*\s+\tx,2.9*\s) {$\boxSet(2)$};
\node[anchor=west,fill=white,inner sep=0pt] at (1.2*\s+\tx,2.4*\s) {$\boxSet(k^*)$};

\node at (2.1,2.2) {$x^*$};
\draw [->] (3.4,2.2) -- (2.4,2.2) node[pos=0.4, fill=white,inner sep=0pt] 	{$\aFun$};
\draw [->] (3.5,1.1) -- (3.5,2.1) node[pos=0.4, fill=white,inner sep=0pt] 	{$\aFun$};
\draw [->] (2.4,0.3) -- (3.4,1.0) node[pos=0.4, fill=white,inner sep=0pt] 	{$\aFun$};
\end{tikzpicture}
\caption{If $\aFun$ is an ACO then the space $\state$ can be divided up into a series of boxes $\boxSet$. Note that this figure is a simplification, as each set $\boxSet(k)$ is decomposable into $\boxSet(k)_1 \times ... \times \boxSet(k)_n$ and so in reality the diagram should be $n$ dimensional.}
\label{fig:aco}
\end{figure}
The reason why these conditions guarantee asynchronous convergence, rather than merely synchronous convergence, is that each box must be decomposable over each of the $n$ nodes. Therefore the operator is always contracting even if every node hasn't performed the same number of updates locally. Note that Theorem~\ref{thm:static-aco-implies-converges} only guarantees $\async$ will converge from states in the initial set $\boxSet(0)$. Hence $\boxSet(0)$ can be thought of as a basin of attraction~\cite{milnor85attractor}.

In practice the set of boxes $\boxSet$ can be difficult and non-intuitive to construct, as they must be explicitly centered around the fixed point whose existence may not even be immediately obvious. \UD{} recognised this and provided several other stronger conditions that are sufficient to construct an ACO. An alternative set of equivalent conditions was originally described by Gurney~\cite{gurney17ultrametrics}. As with the ACO conditions, these conditions were relaxed by~\cite{anonymised} and the latter version is now presented.

\begin{definition}[Static AMCO]
A function $F$ is an asynchronously metrically contracting operator (AMCO) if for every node $i$ there exists a distance function $d_i$ such that if $D(x,y) \triangleq \max_{i} d_i(x_{i}, y_{i})$ then:
\begin{enumerate}[label=(SU\arabic*),leftmargin=1.5cm]
\item $\forall i \in \nodeSet, x, y \in \state : d_i(x,y) = 0 \Leftrightarrow x = y$
\label{ass:static-ultra-metric}
\item $\forall i \in \nodeSet : \exists n \in \nat : \forall x, y \in \state : d_i(x,y) \leq n$
\label{ass:static-ultra-bounded}
\item $\forall x \in \state : x \neq \aFun(x) \Rightarrow D(x,\aFun(x)) > D(\aFun(x),\aFun^2(x))$
\label{ass:static-ultra-contr-orbits}
\item $\forall x, x^* \in \state: \aFun(x^*) = x^* \wedge x \neq x^* \Rightarrow D(x^*,x) > D(x^*,\aFun(x))$
\label{ass:static-ultra-contr-fixed}
\item $\state$ is non-empty
\label{ass:static-ultra-nonempty}
\end{enumerate}
\end{definition}
The AMCO conditions require the construction of a notion of distance between states such that there exists a maximum distance \ref{ass:static-ultra-bounded} and that successive iterations become both closer together \ref{ass:static-ultra-contr-orbits} and closer to any fixed point \ref{ass:static-ultra-contr-fixed}. Note, unlike Gurney's original formulation, the AMCO conditions as defined above do not require $d_i$ to obey the typical metric axioms of symmetry and the ultrametric triangle inequality.

Gurney~\cite{gurney17ultrametrics} proves that the AMCO conditions are equivalent to the ACO conditions by constructing reductions in both directions. Consequently the following convergence theorem holds.
\begin{theorem}
\label{thm:metric_async_convergence}
If $F$ is an AMCO then $\async$ converges deterministically over the set~$\state$.
\end{theorem}
\begin{proof}
See  \cite{gurney17ultrametrics} \& \cite{anonymised}. 
\end{proof}

\subsection{Motivations for a dynamic model}
\label{sec:static-drawbacks}

As discussed in the introduction, prior work applying \UD's results to ``always-on" algorithms often assumes that dynamic iterations can be viewed as a sequence of static iterations. By inspecting the definition of $\delta$, the flaw in this argument can now be formalised. Consider a dynamic iteration with nodes $\nodeSet$ in which node $i \in \nodeSet$ has sent out an update message to $j \in \nodeSet$ and then $i$ ceases participating. The new static iteration would begin immediately with participants $\nodeSet - \{ i \}$ and therefore when $j$ next activates, the static model is incapable of receiving the message from node~$i$.

Another feature lacking in the static model is the ability to reboot nodes. It is possible to represent temporary node failure in the static model by excluding it from the set of active nodes, however this still provides an unsatisfactory model as many types of failure will result in a node's state being erased (e.g.~replacing a faulty server in a data centre). In reality after such an event the node is forced to revert back to the initial state. This ``rebooting'' of a node after a temporary failure cannot be described by the existing static model. 
\section{Dynamic asynchronous iterations}
\label{sec:dynamic-model}

To overcome these shortcomings we now propose a new, more general model that can describe both dynamic and static iterations.

\subsection{Model}

Let $\nodeSet$ be the set of all the nodes that participate at some point during the dynamic iteration. $\nodeSet$ is still assumed to be finite with $n = |V|$, as the only cases in which $|V|$ could be infinite is if either an infinite number of nodes participated at the same time or an infinite amount of time has passed since the iteration began. Neither case is useful in reality. As before, we assume there exists a product state space $\state = \state_1 \times \state_2 \times .. \times \state_n$.

 In order to capture the new dynamic nature of the iteration we introduce the concept of an \emph{epoch}. An epoch is a contiguous period of time in which both the function being iterated and the set of participating nodes remain constant. The set of epochs is denoted as $\epochSet$ but as with time can be assumed to be an alias for $\mathbb{N}$.

Instead of a single function $\aFun$, we now assume that $\aFun$ is a family of indexed functions where $\aFun^{ep}$ is the function being computed in epoch $e \in \epochSet$ by participants $p \subseteq \nodeSet$. Furthermore we assume there exists a special non-participating state $\bot \in \state$.

A schedule must therefore not only track the activation of nodes and the flow of data between them but also the current epoch and the participating nodes. Given these requirements it is natural to redefine a schedule as follows:

\begin{definition}[Dynamic schedule]
A dynamic schedule is a tuple of functions $(\alpha, \beta, \eta, \pi)$ where:
\begin{itemize}
\item $\alpha : \timeSet \rightarrow 2^\nodeSet$ is the \emph{activation function}, where $\alpha(t)$ is the set of nodes which activate at time~$t$.

\item $\beta : \timeSet \times \nodeSet \times \nodeSet \rightarrow \timeSet$ is the \emph{data flow function}, where $\beta(t,i,j)$ is the time at which the information used by node $i$ at time $t$ was sent by node $j$.

\item $\eta: \timeSet \rightarrow \epochSet$ is the \emph{epoch function}, where $\eta(t)$ is the epoch at time $t$.

\item $\pi: \epochSet \rightarrow 2^\nodeSet$ is the \emph{participants function}, where $\pi(e)$ is the set of nodes participating in the computation during epoch $e$.
\end{itemize}
such that:
\begin{enumerate}[label=(DS\arabic*),leftmargin=1.5cm]
\item $\forall i,j,t : \: \beta(t+1,i,j) \leq t$ -- information only travels forward in time.
\label{ass:dynamic-sched-causality}
\item $\forall t_1,t_2 : \: t_1 \leq t_2 \Rightarrow \eta(t_1) \leq \eta(t_2)$ -- the epoch number only increases.
\label{ass:dynamic-sched-monotonicity}
\end{enumerate}
\end{definition}
\vspace{0.5em}
The additional assumption \ref{ass:dynamic-sched-monotonicity} states that epochs are monotonically increasing. Although not technically required, the assumption is convenient as it ensures that for any two points in time in the same epoch then every point between them is also in the same epoch. This assumption does not reduce the expressive power of the model, as for any non-monotonic $\eta$ it is possible to find a suitable relabelling of epochs that recovers monotonicity. Another possible assumption that might be made is that a node can only activate if it is currently participating in the iteration (i.e. $\forall t : \alpha(t) \subseteq \pi(\eta(t))$). Although the assumption is reasonable, the dynamic asynchronous state function $\async$ will be defined in such a way that it will not be required (see Definition~\ref{def:dynamic-delta}).

Given a schedule, we define some additional notation for $\rho(t)$, the set of nodes participating at time $t$, and $\aFun^t$, the function being used at time $t$:
\begin{align*}
\rho(t) & \triangleq \pi(\eta(t)) \\
F^t 	& \triangleq F^{\eta(t)\rho(t)}
\end{align*}
It is now possible to define the \emph{dynamic} asynchronous state function as follows:
\begin{definition}[Dynamic asynchronous state function]
Given an initial state~$x$ and a schedule $(\alpha, \beta, \eta, \pi)$ the dynamic state function is defined as:
\begin{equation*}
\label{def:dynamic-delta}
\async^t_i(x) = \begin{cases}
\bot_i & \text{if $i \notin \rho(t)$} \\
x_i & \text{else if $t = 0$ or $i \notin \rho(t-1)$}\\
\async^{t-1}_i(x) & \text{else if $i \notin \alpha(t)$} \\
\aFun^t_i(\async^{\beta(t,i,1)}_1(x), \ldots , \async^{\beta(t,i,n)}_n(x)) & \text{otherwise}
\end{cases}
\end{equation*}
where $\async^t_i(x)$ is the state of node $i$ at time $t$ starting from state $x$.
\end{definition}
If a node is not currently participating then it adopts its non-participating state. If it is participating at time $t$ but was not participating at the time $t-1$ then it must have just (re)joined the computation and it therefore adopts its initial state. If the node is a continuing participant and is inactive at time~$t$ then its state remains unchanged. Otherwise, if it is active at time $t$, it updates its state in accordance with the data received from the other nodes in the computation. 

Note that in the latter case at time $t$ nodes can use data from any node in $V$ rather than just the current set of participants $\rho(t)$. Hence nodes that are currently participating may end up processing messages from nodes that are no longer participating in the current epoch. Also note that this new model is a strict generalisation of the static model as the static definition of $\delta$ is immediately recovered by a schedule with the constant epoch and participants functions $\eta(t) = 0$ and $\pi(0) = V$.

\subsection{Correctness}

In order to define a notion of correctness for dynamic iterations, we first need to update the definition of a pseudocycle. It turns out that only two alterations are needed. The first is that the start and end of activation and expiry periods and consequently pseudocycles must belong to the same epoch. The second is that during a pseudocycle, only the participating nodes need to experience an activation and expiry period. An updated version of the definitions is given below with the changes underlined.

\begin{definition}[Dynamic activation period]
A period of time $[t_1,t_2]$ is a \emph{dynamic activation period} for node~$i$ if \underline{$\eta(t_1) = \eta(t_2)$} and there exists a time $t \in [t_1,t_2]$ such that $i \in \alpha(t)$.
\end{definition}

\begin{definition}[Dynamic expiry period]
A period of time $[t_1,t_2]$ is\ref{ass:static-ultra-metric} a \emph{dynamic expiry period} for node~$i$ if \underline{$\eta(t_1) = \eta(t_2)$} and for all nodes $j$ and times $t \geq t_2$ then $t_1 \leq \beta(t,i,j)$.
\end{definition}

\begin{definition}[Dynamic pseudocycle]
A period of time $[t_1,t_2]$ is a \emph{dynamic pseudocycle} if \underline{$\eta(t_1) = \eta(t_2)$} and \underline{for all nodes $i$ $\in \rho(t_1)$} there exists a time $t \in [t_1,t_2]$ such that $[t_1,t]$ is an expiry period for node $i$ and $[t,t_2]$ is an activation period for node $i$.
\end{definition}

We can now start to think what it means for a dynamic iteration to be implemented correctly. Guaranteeing that a dynamic iteration will always converge to one particular fixed point is impossible as both the underlying computation and the participants may continue to change indefinitely. Furthermore the epoch durations may be short enough that no fixed point is ever reached, even temporarily. The natural and intuitive notion in such circumstances is to say that an iteration is \emph{convergent} if whenever an epoch contains a sufficient number of pseudocycles then $\async$ will converge to a fixed point for the remainder of that epoch. Furthermore within an epoch the fixed point reached should be unique, but different epochs may have different unique fixed points.

However, we would also like to be able to reason about for which epochs and sets of participants a dynamic iteration does \emph{not} converge. For example in the case of inter-domain routing, it is known that path-vector protocols only converge if and only if the network topology is free~\cite{sobrinho17correctness}. Therefore, in the same way that we constrain convergence to some set of initial states, we also constrain convergence to some set of pairs of epochs and set of participants. We will refer to such pairs as configurations.

\begin{definition}[Convergent iteration]
A dynamic asynchronous iteration is convergent over an initial set of states ${\initialSet = \initialSet_1 \times \initialSet_2 \times \ldots \times \initialSet_n}$ and a set of configurations $\configSet \subseteq \epochSet \times 2^\nodeSet$ iff:
\begin{enumerate}
\item for every epoch and configuration $(e ,p) \in \configSet$ there exists a fixed point~$x^*_{ep}$ for $\aFun^{ep}$ and a number of iterations~$k^*_{ep}$.
\item for every initial state $x \in X$, schedule and time $t_1$ then if ${(\eta(t_1), \rho(t_1)) \in \configSet}$ and the time period $[t_1,t_2]$ contains $k^*_{\eta(t_1)\rho(t_1)}$ pseudocycles then for every time~$t_3$ such that $t_3 \geq t_2$ and $\eta(t_2) = \eta(t_3)$ then $\async^{t_3}(x) = x^*_{\eta(t_1)\rho(t_1)}$.
\end{enumerate}
\end{definition}

Having now defined what we mean for a dynamic iteration to be correct, in the next section we generalise the static ACO and AMCO conditions described in Section~\ref{sec:static-results} and prove analogous correctness theorems for them.

\section{Results}
\label{sec:convergent-results}

Before we generalise the ACO and AMCO conditions, some additional definitions are needed. As would be expected, information from non-participating nodes that is still ``in-flight'' from a previous epoch may interfere with the convergence of $\delta$ in the current epoch. Therefore a notion is needed of a state only containing information for the current set of participants. 
\begin{definition}[Accordant states]
A state $x$ is \emph{\accordant{}} with respect to a set of participants $p$ if every inactive node is assigned the inactive state, i.e. $\forall i \notin p : x_i = \bot_i$. The set of states that are \accordant{} with $p$ is denoted as $\accSet$.
\end{definition}
When in the dynamic world we also need to take more care about the properties of $\initialSet$, the set of initial states which the iteration can converge from. The static ACO conditions in Definition~\ref{def:static-aco} implicitly take the first box to be the set of initial states, i.e. $X = B(0)$, whilst the static AMCO conditions implicitly assume any state is valid, i.e. $\initialSet = \state$. However, the former approach no longer works in the dynamic world as we are forced to have different sets of boxes for each epoch and set of participants, and the latter is unnecessarily restrictive, as some iterative algorithms may only converge when started from certain states. 

In order to solve these problems,  we now define the properties that the initial set must satisfy regardless of whether we're using the ACO or AMCO conditions.
\begin{definition}[Valid set of initial states]
An initial set $\initialSet = \initialSet_1 \times \initialSet_2 \times \ldots \times \initialSet_n$ is \emph{valid} if:
\begin{enumerate}
[label=(IS\arabic*),leftmargin=1.5cm]
\item $\bot \in \initialSet$
\label{ass:initial-set-invalid}

\item $\forall e, p, x : x \in \initialSet \Rightarrow F^{ep}(x) \in \initialSet$
\label{ass:initial-set-closed}
\end{enumerate}
\end{definition}

\noindent Assumption~\ref{ass:initial-set-invalid} states that the non-participating state is in the initial set and \ref{ass:initial-set-closed} states that $\initialSet$ is closed over every operator. The latter is the counterpart of assumption~\ref{ass:static-aco-closed} in the definition of a static ACO. Together these ensure that an asynchronous iteration never leaves the initial set (see Lemma~\ref{lem:initial-set} for details). Also note that the entire state space $\state$ is trivially a valid initial set.

\subsection{Dynamic ACO implies convergent}
\label{sec:aco-implies-convergent}

We can now define the dynamic counterpart of the static ACO conditions. While it might be tempting to simply require that every $F^{ep}$ be a static ACO, there are a couple of additional constraints necessary.

\begin{definition}[Dynamic ACO]
\label{def:dynamic-aco}
The set of functions $F$ are a \emph{dynamic ACO} over a set of initial states $\initialSet = \initialSet_1 \times \initialSet_2 \times \ldots \times \initialSet_n$ and set of configurations $\configSet \subseteq \epochSet \times 2^\nodeSet$ if for every epoch~$e$ and set of participants~$p$ such that $(e , p) \in \configSet$ there exists a sequence of sets $\boxSet^{ep}(k) = \boxSet^{ep}_{1}(k) \times \boxSet^{ep}_{2}(k) \times ... \times \boxSet^{ep}_{n}(k)$ for $k \in \mathbb{N}$ such that:
\begin{enumerate}[label=(DA\arabic*),leftmargin=1.5cm]
\item $\initialSet \subseteq \boxSet^{ep}(0)$
\label{ass:dynamic-aco-first-box-eq}

\item $\forall k \in \nat, i \notin p : \bot_i \in \boxSet^{ep}_{i}(k)$
\label{ass:dynamic-aco-bot-set}

\item $\forall k \in \nat, x \in \initialSet \cap \accSet: x \in \boxSet^{ep}(k) \Rightarrow F^{ep}(x) \in \boxSet^{ep}(k+1)$
\label{ass:dynamic-aco-progress}

\item $\exists k^*_{ep} \in \nat,x^*_{ep} \in \initialSet: \forall k \in \nat: k^*_{ep} \leq k \Rightarrow \boxSet^{ep}(k) = \{x^*_{ep}\}$
\label{ass:dynamic-aco-finish}
\end{enumerate}
\end{definition}
Assumption~\ref{ass:dynamic-aco-first-box-eq} is a new assumption that links the initial boxes of each epoch together by assuming that the initial set of states is a subset of the initial box. Assumption~\ref{ass:dynamic-aco-bot-set} is also new and ensures that the box for any non-participating node contains its non-participating state. Assumptions \ref{ass:dynamic-aco-progress} \& \ref{ass:dynamic-aco-finish} are generalised versions of \ref{ass:static-aco-progress} \& \ref{ass:static-aco-finish} respectively. The only difference is that \ref{ass:dynamic-aco-progress} has been weakened so that applying $\aFun^{ep}$ only advances a box when the state is \accordant{} with the current set of participants. This means that progress need not be made in the case when stale messages are still being received from nodes that are no longer participating.

We now prove that if $\aFun$ is a dynamic ACO over a valid set of initial states~$\initialSet$ and configurations $\configSet$ then $\async$ is convergent over $\initialSet$ and $\configSet$. Going forwards the existence of some arbitrary schedule $(\alpha,\beta,\eta,\pi)$ and starting state $x \in \initialSet$ is assumed. As with $F^t \triangleq F^{\eta(t)\rho(t)}$, we use the shorthand $\boxSet^t \triangleq \boxSet^{\eta(t)\rho(t)}$ and $c(t) \triangleq (\eta(t), \rho(t))$ so that the current boxes and configuration may be indexed by time rather than by epoch and participants. Initially some auxiliary definitions are introduced in order to improve the readability of the proof.
\begin{definition}
\emph{The state of node $i$ is in box $k$ at time $t$} if:
\begin{equation*}
c(t) \in \configSet \Rightarrow \async^t_i(x) \in \boxSet^{t}_i(k)
\end{equation*}
\end{definition}
\noindent i.e. if the current configuration is in the set of valid configurations then the current state of node~$i$ is in box~$k$.

\begin{definition}
\label{def:node_messages_in_box}
\emph{The messages to node $i$ are in box $k$ at time $t$} if:
\begin{equation*}
c(t) \in \configSet \Rightarrow \forall s : (s > t) \wedge (\eta(s) = \eta(t)) \Rightarrow \forall j: \async^{\beta(s,i,j)}_j(x) \in \boxSet^{t}_j(k)
\end{equation*}
\end{definition}
\noindent i.e. if the current configuration is in the set of valid configurations then any message arriving at node $i$ after time $t$ and before the end of the current epoch is guaranteed to be in box $k$. A different way of viewing this condition is that node $i$'s local view of the global state of the iteration will be in box~$k$ for the remainder of the epoch.

\begin{definition}
\emph{The messages to node $i$ are accordant at time $t$} if:
\begin{equation*}
\forall s : (s > t) \wedge (\eta(s) = \eta(t)) \Rightarrow \forall j: j \notin \rho(s) \Rightarrow \async^{\beta(s,i,j)}_j = \bot_j
\end{equation*}
\end{definition}
\noindent i.e. any message arriving at node $i$ after time $t$ during the current epoch from a non-participating node $j$ will always contain the non-participating state $\bot_j$. This is equivalent to stating that node $i$'s local view of the state is accordant.

\begin{definition}
\label{def:node_computation_in_box}
\emph{The computation at node $i$ is in box $k$ at time $t$} if:
\begin{enumerate}
\item the state of node $i$ is in box $k$ at time $t$.
\item if $k > 0$ the messages to node $i$ are in box $k - 1$ at time $t$.
\item if $k > 0$ then the messages to node $i$ are accordant at time $t$.
\end{enumerate}
\end{definition}
\noindent This definition collects together the pre-conditions required to prove that the state of node $i$ will always be in box $k$ for the remainder of the epoch, as shown in Lemma~\ref{lem:persist-state}. Finally we lift this definition from an individual node to the whole computation as follows:
\begin{definition}
\label{def:computation_in_box}
\emph{The computation is in box $k$ at time $t$} if for all nodes $i \in \rho(t)$ then the computation at node $i$ is in box $k$ at time $t$.
\end{definition}
\noindent It is interesting to note that Definition~\ref{def:computation_in_box} does not place any requirements on non-participating nodes. This is because, by the definition of $\async$, any non-participating node $i$ is always in the non-participating state $\bot_i$, which, by assumption~\ref{ass:dynamic-aco-bot-set}, is in every one of the boxes, including the final one. Also note that all of the above definitions contain some linguistic slight of hand, as being in box~$k$ at time~$t$ and being in box~$k$ at time~$t+1$, does not necessarily refer to the same box if $\eta(t) \neq \eta(t+1)$.

The proof can now be split into four parts. The first set of \emph{closure} lemmas prove that the computation is always in box 0 even after changes in the epoch. The second set of \emph{stability} lemmas describe under what conditions after the computation reaches box $k$ it remains in that box for the remainder of the epoch. The third set of \emph{progress} lemmas demonstrate how during a pseudocycle the entire computation advances at least one box. Finally these results are combined to prove convergence.

\subsubsection{Closure lemmas}

In order to later apply the other ACO assumptions, we first establish that the initial set $\initialSet$ is closed over $\delta$, i.e. that the iteration never escapes the initial set. As a consequence of this and assumption~\ref{ass:dynamic-aco-first-box-eq}, we then prove that both the state and the computation are always in the box 0 of the current epoch.

\begin{lemma}
\label{lem:initial-set}
For any $x \in \initialSet$ and time $t$ then $\async^t(x) \in \initialSet$.
\end{lemma}
\begin{proof}
Consider an arbitrary node $i$. The proof that $\async^t_i(x) \in \initialSet_i$ proceeds by induction over the definition of $\async$. 

\case{1}{$i \notin \rho(t)$}
Then $\async^{t}_i(x) = \bot_i$ and $\bot_i \in \initialSet_i$ by assumption~\ref{ass:initial-set-invalid}.

\case{2}{$i \in \rho(t)$ and ($t = 0$ or $i \notin \rho(t-1)$)}
Then $\async^t_i(x) = x_i$ and $x_i \in \initialSet_i$ by the initial assumption.

\case{3}{$i \in \rho(t)$ and $i \in \rho(t-1)$ and $i \notin \alpha(t_2)$}
Then $\async^{t}_i(x) = \async^{t-1}_i(x)$, and $\async^{t-1}_i(x) \in \initialSet_i$ by the inductive hypothesis applied to time~$t-1$.

\case{4}{$i \in \rho(t)$ and $i \in \rho(t-1)$ and $i \in \alpha(t)$}
Then $\async^{t}_i(x) = \aFun^t_i(\async^{\beta(t,i,1)}_1(x), \ldots, \async^{\beta(t,i,n)}_n(x))$. For each $j$ then $\async^{\beta(t,i,j)}_j(x) \in \initialSet_j$ by the inductive hypothesis applied to time $\beta(t,i,j)$. Hence 
$\aFun^t_i(...) \in \initialSet_i$ as $\initialSet$ is closed under $\aFun^t$ by assumption \ref{ass:initial-set-closed} .
\end{proof}

\begin{lemma}
\label{lem:initial-state}
For every time $t$ and node $i$ the state of node $i$ is in box $0$ at time~$t$.
\end{lemma}
\begin{proof}
Consider an arbitrary time $t$ and node $i$ such that $c(t) \in \configSet$. Then $\async^t_i(x) \in \initialSet_i$ by Lemma~\ref{lem:initial-set} and $\initialSet_i \subseteq \boxSet^t_i(0)$ by assumption~\ref{ass:dynamic-aco-first-box-eq}.
\end{proof}

\begin{lemma}
\label{lem:initial-computation}
For every time $t$ and node $i$ the computation at node $i$ is in box $0$ at time~$t$.
\end{lemma}
\begin{proof}
The state of node $i$ is in box $0$ at time $t$ by Lemma~\ref{lem:initial-state}. As we are only considering box 0, Definition~\ref{def:node_computation_in_box} does not require us to prove that the messages to node $i$ are in box 0 and are accordant at time $t$\footnote{Note that the latter isn't even true as out-of-date messages may still be arriving from nodes that were participating in a previous epoch but are no longer participating in the current epoch.}.
\end{proof}

\subsubsection{Stability lemmas}

Guaranteeing that the dynamic iteration makes progress towards the fixed point of the current epoch is complicated by the fact that out-of-date messages may arrive from earlier in the iteration and undo recent progress. The next lemmas establish what conditions are necessary to guarantee that once the state and messages are in box~$k$ then they will remain in box~$k$ for the remainder of the epoch. 

\begin{lemma}
\label{lem:persist-state}
If the computation at node $i$ is in box $k$ at time $t$ then the state of node $i$ is in box $k$ for every time $s \geq t$ such that $\eta(s) = \eta(t)$.
\end{lemma}
\begin{proof}
Assume that the computation at node $i$ is box $k$ at time $t$ for an arbitrary node~$i$ and time~$t$. We must show that $\async^{s}_i(x) \in \boxSet^s_i(k)$ for any $s \geq t$ such that $c(s) \in \configSet$. If $k = 0$ then the result follows immediately by Lemma~\ref{lem:initial-state}. Otherwise if $k > 0$ the proof proceeds by induction over time $s$ and the definition of~$\async$. If $s = t$ then the state of node~$i$ is in box~$k$ at time~$t$ by Definition~\ref{def:node_computation_in_box}. Otherwise $s > t$ and as $s-1 \in [t , s]$ and $\eta(t) = \eta(s)$ then $\eta(t) = \eta(s-1) = \eta(s)$ and hence $\boxSet^t = \boxSet^{s-1} = \boxSet^{s}$ and $c(t), c(s-1) \in \configSet$. Consider the following cases:

\case{1}{$i \notin \rho(s)$}
Then $\async^{s}_i(x) = \bot_i$ and $\bot_i \in \boxSet^s_i(k)$ by assumption~\ref{ass:dynamic-aco-bot-set}.

\case{2}{$i \in \rho(s)$ and $i \notin \rho(s-1)$}
As $\eta(s - 1) = \eta(s)$ then $\rho(s-1) = \rho(s)$, contradicting the case assumptions.

\case{3}{$i \in \rho(s)$ and $i \in \rho(s-1)$ and $i \notin \alpha(s)$}
Then $\async^{s}_i(x) = \async^{s-1}_i(x)$. As $c(s-1) \in \configSet$ then we have $\async^{s-1}_i(x) \in \boxSet^{s-1}_i(k)$ by the inductive hypothesis applied to time $s-1$. As $\boxSet^{s-1}_i(k) = \boxSet^{s}_i(k)$, we therefore have $\async^{s}_i(x) \in \boxSet^{s}_i(k)$.

\case{4}{$i \in \rho(s)$ and $i \in \rho(s-1)$ and $i \in \alpha(s)$}
Then $\async^{s}_i(x) = \aFun^s_i(\async^{\beta(s,i,1)}_1(x), \ldots, \async^{\beta(s,i,n)}_n(x))$. As $c(t) \in \configSet$ and all messages to node $i$ are in box $k-1$ at time $t$ and are accordant, then $\async^{\beta(s,i,j)}_j(x)$ is accordant and in box $\boxSet^{t}_j(k-1) = \boxSet^s_j(k-1)$ for every node $j$. Hence $\aFun^s_i(...) \in \boxSet^s_i(k)$ by assumption~\ref{ass:dynamic-aco-progress}.
\end{proof}

\begin{lemma}
\label{lem:persist-messages}
If messages to node $i$ are in box $k$ at time $t$ then the messages to node $i$ are in box $k$ for all times $s \geq t$ such that $\eta(s) = \eta(t)$.
\end{lemma}
\begin{proof}
Consider a time $r > s$ such that $\eta(r) = \eta(s)$. We must show that $\async^{\beta(r,i,j)}_j(x) \in \boxSet^s_j(k)$ for every node $j$. Then $r > t$ and $\eta(r) = \eta(t)$ and so by Definition~\ref{def:node_messages_in_box} we have that $\async^\beta(r,i,j)(x)_j \in \boxSet^t_j(k)$. As $\eta(s) = \eta(t)$ then $\boxSet^s = \boxSet^t$ and hence we have the required result.
\end{proof}

\subsubsection{Progress lemmas}

Having established that i) the computation is always in box 0 no matter the epoch and ii) once the computation at node~$i$ has reached box~$k$, it remains in box~$k$, it is next necessary to establish when the computation advances a box during an epoch. These conditions are intimately tied to the notion of a pseudocycle.

\begin{lemma}
\label{lem:advance-state}
If the messages to node $i$ are accordant and are in box $k$ at time~$t$ and $[t,s]$ is an activation period then the state of node $i$ is in box $k+1$ at time~$s$.
\end{lemma}
\begin{proof}
Assume $c(s) \in \configSet$. The proof that $\async^s_i(x) \in \boxSet^s_i(k+1)$ proceeds by induction over the definition of $\delta$ and time $s$. As activation periods are of non-zero length then $s > t$ and as $s-1 \in [t , s]$ and $\eta(t) = \eta(s)$ then $\boxSet^t = \boxSet^{s-1} = \boxSet^{s}$ and $c(t), c(s-1) \in \configSet$. Consider the following cases:

\case{1}{$i \notin \rho(s)$}
Then $\async^{s}_i(x) = \bot_i$ and $\bot_i \in \boxSet^s_i(k+1)$ by assumption~\ref{ass:dynamic-aco-bot-set}.

\case{2}{$i \in \rho(s)$ and $i \notin \rho(s-1)$}
As $\eta(s - 1) = \eta(s)$ then $\rho(s-1) = \rho(s)$, contradicting the case assumptions.

\case{3}{$i \in \rho(s)$ and $i \in \rho(s-1)$ and $i \notin \alpha(s)$}
Then $\async^{s}_i(x) = \async^{s-1}_i(x)$. If $s-1 = t$ then the initial assumptions are contradicted as $i$ has not activated during the period $[t,s]$. Otherwise if $s-1 > t$  then $\async^{s-1}_i(x) \in \boxSet^{s-1}_i(k+1)$ by applying the inductive hypothesis to time $s-1$. As $\boxSet^{s-1}_i(k+1) = \boxSet^s_i(k+1)$ we have the required result.

\case{4}{$i \in \rho(s)$ and $i \in \rho(s-1)$ and $i \in \alpha(s)$}
Then $\async^{s}_i(x) = \aFun_i(\async^{\beta(s,i,1)}_1(x), \ldots, \async^{\beta(s,i,n)}_n(x))$. As $c(t) \in \configSet$ and all messages to node $i$ are in box $k$ at time $t$ and are accordant, then $\async^{\beta(s,i,j)}_j(x)$ is accordant and in box $\boxSet^{t}_j(k) = \boxSet^s_j(k)$ for every node $j$. Hence $\aFun_i(...) \in \boxSet^s_i(k+1)$ by assumption~\ref{ass:dynamic-aco-progress}.
\end{proof}

\begin{lemma}
\label{lem:advance-messages}
If the computation is in box $k$ at time $t$ and $[t,s]$ is an expiry period for node~$i$ then the messages to node $i$ are in box $k$ at time $s$.
\end{lemma}
\begin{proof}
Assume that the computation is in box $k$ at time $t$ and consider two arbitrary nodes~$i$~and~$j$ and time $s$ such that $[t , s]$ is an expiry period and $c(s) \in \configSet$. We must show that for all times $r > s$ such that $\eta(s) = \eta(r)$ then $\async^{\beta(r,i,j)}_j(x) \in \boxSet^s_j(k)$. As $[t,s]$ is an expiry period then $t \leq \beta(r,i,j) < r$ and therefore $\eta(t) = \eta(\beta(r,i,j)) = \eta(r) = \eta(s)$. If $j \notin \rho(s)$ then $\async^{\beta(r,i,j)}_j(x) = \bot_j$ and $\bot_j \in  \boxSet^s_j(k)$ by assumption~\ref{ass:dynamic-aco-bot-set}. Otherwise if $j \in \rho(s)$ then $\async^{\beta(r,i,j)}_j(x) \in \boxSet^{\beta(r,i,j)}_j(k)$ by Lemma~\ref{lem:persist-state} applied to time period $[t, \beta(r,i,j)]$ and the fact that the computation at node~$j$ is in box~$k$ at time~$t$. The required result then follows as $\boxSet^{\beta(r,i,j)} = \boxSet^s$ by $\eta(\beta(r, i, j)) = \eta(s)$.
\end{proof}

\noindent Lemmas~\ref{lem:advance-state}~\&~\ref{lem:advance-messages} prove that during activation and expiry periods the state and the messages are both guaranteed to advance at least one box. The next lemma combines them to prove that during a pseudocycle the whole computation advances at least one box, i.e. during a pseudocycle the asynchronous iteration makes at least as much progress as a single step of the synchronous iteration.

\begin{lemma}
\label{lem:pseudocycle}
If the computation is in box $k$ at time $t$ and the period $[ t , s ]$ is a pseudocycle then the computation is in box $k+1$ at time $s$.
\end{lemma}
\begin{proof}
Consider an arbitrary node $i \in \rho(t)$. As $[t , s]$ is a pseudocycle then as $i \in \rho(t)$ there exists a time $r$ such that $[t,r]$ is an expiry period for node $i$ and $[r,s]$ is an activation period for node $i$.
\begin{itemize}
\item As the messages to node $i$ are accordant at time $t$ then they are also accordant at times $r$ and $s$.

\item As $[t,r]$ is an expiry period and the computation is in box $k$ at time $t$, then the messages to node $i$ are in box $k$ at time $r$ by Lemma~\ref{lem:advance-messages}, and also therefore at time $s$ by Lemma~\ref{lem:persist-messages}.

\item As $[r,s]$ is an activation period and the messages to node $i$ are accordant and in box $k$ at time $r$ (by the previous two points) then the state of node~$i$ in box~$k+1$ at time $s$ by Lemma~\ref{lem:advance-state}.
\end{itemize}
Consequently all three requirements for the computation at node $i$ being in box $k+1$ at time $s$ are fulfilled.
\end{proof}

\subsubsection{Convergence}

Now that Lemma~\ref{lem:pseudocycle} has established that during a pseudocycle the whole computation advances one box, the main theorem may be proved.

\begin{theorem}
\label{thm:aco-implies-convergent}
If $\aFun$ is a dynamic ACO over a valid initial set $\initialSet$ and configurations $\configSet$ then $\async$ is convergent over $\initialSet$ and $\configSet$.
\end{theorem}
\begin{proof}
To prove that $\async$ is convergent it is first necessary to construct a fixed point $x^*_{ep}$ and iteration number $k^*_{ep}$ for every epoch $e$ and set of participants~$p$ such that $(e,p) \in \configSet$. Let these be the~$x^*_{ep}$ and~$k^*_{ep}$ respectively as specified by assumption~\ref{ass:dynamic-aco-finish}.

Next consider an arbitrary schedule, starting state $x \in \initialSet$ and starting time~$t_1$ in epoch $e = \eta(t_1)$ with participants $p = \rho(t_1)$ such that $(e, p) \in \configSet$.  We must show that for all times $t_2$ if $[t_1,t_2]$ contains $k^*_{ep}$ pseudocycles then for all times $t_3$ such that $t_3 \geq t_2$ and $\eta(t_3) = \eta(t_2)$ then $\async^t_3(x) = x^*_{ep}$. 

The computation is always in box 0 by Lemma~\ref{lem:initial-computation}. Consequently after $k^*_{ep}$ pseudocycles, the computation is in box~$k^*_{ep}$ at time $t_2$ by repeated application of Lemma~\ref{lem:pseudocycle}. Hence for any subsequent time~$t_3$ in epoch $e$, then $\async^{t_3}(x) \in \boxSet^{ep}(k^*_{ep})$ by Lemma~\ref{lem:persist-state} and, as $x^*_{ep}$ is the only state in $\boxSet^{ep}(k^*_{ep})$ by assumption~\ref{ass:dynamic-aco-finish}, then $\async^{t_3}(x) = x^*_{ep}$.
\end{proof}

\subsection{Dynamic AMCO implies convergent}

Although the dynamic ACO conditions are sufficient to guarantee convergence, they can be a tricky to construct in practice. As discussed previously in Section~\ref{sec:static-results}, the AMCO conditions are often easier to work with. This section defines the dynamic AMCO conditions and shows that they also guarantee the iteration is convergent by constructing a reduction from the dynamic AMCO conditions to the dynamic ACO conditions. 

\begin{definition}[Dynamic AMCO]
The set of functions $F$ are a \emph{dynamic AMCO} over a set of initial states $\initialSet = \initialSet_1 \times \initialSet_2 \times ... \times \initialSet_n$ and a set of configurations $\configSet \subseteq \epochSet \times 2^\nodeSet$ if for every epoch~$e$ and set of participants $p$ such that $(e, p) \in \configSet$ and for every node $i \in \nodeSet$ there exists a distance function $d^{ep}_i$ such that if $D^{ep}(x,y) \triangleq \max_{i \in p} \: d^{ep}_i(x,y)$ then:
\begin{enumerate}[resume, label=(DU\arabic*),leftmargin=1.5cm]
\item $\forall i \in \nodeSet: \forall x,y \in S: d^{ep}_i(x,y) = 0 \Leftrightarrow x = y$
\label{ass:dynamic-ultra-eq}
\item $\forall i \in \nodeSet: \exists n : \forall x,y \in S : d^{ep}(x,y)_i \leq n$
\label{ass:dynamic-ultra-bounded}
\item $\forall x {\in} \initialSet {\cap} \accSet : \aFun^{ep}(x) \neq x \Rightarrow D^{ep}(\aFun^{ep}(x),(\aFun^{ep})^2(x)) {<} D^{ep}(x,\aFun^{ep}(x))$
\label{ass:dynamic-ultra-contr-orbits}
\item $\forall x {\in} \initialSet {\cap} \accSet, x^* {\in} \state{:} \aFun^{ep}(x^*) {=} x^* \wedge x {\neq} x^* {\Rightarrow} D^{ep}(x^*,\aFun^{ep}(x)) {<} D^{ep}(x^*,x)$
\label{ass:dynamic-ultra-contr-fixed}
\item $\forall x {\in} \initialSet {\cap} \accSet: F^{ep}(x) \in \accSet$.
\label{ass:dynamic-ultra-not-participating}
\end{enumerate}
\end{definition}
\noindent Again assumptions $(DU1)$ -- $(DU4)$ are generalisations of $(SU1)$ -- $(SU4)$. The crucial difference is that everything is restricted to the set of participants: ~$F^{ep}$ need only be strictly contracting over \accordant{} states $\accSet$, and the distance functions $D^{ep}$ are defined as the maximum over all participating states. Note that the static assumption \ref{ass:static-ultra-nonempty} that $\state$ is non-empty is not needed as the dynamic model assumes the existence of the non-participating state $\bot \in \state$. Instead the new assumption~\ref{ass:dynamic-ultra-not-participating} ensures that the operator $\aFun$ preserves accordant enforces that non-participating nodes adopt the non-participating state. This assumption was not stated explicitly in the dynamic ACO conditions but can be derived from assumptions \ref{ass:dynamic-aco-bot-set} and \ref{ass:dynamic-aco-progress}.

The proof that these conditions imply that the iteration is convergent is a generalisation of the proof in~\cite{anonymised} which in turn was based off the work in~\cite{gurney17ultrametrics}. The main thrust of the reduction to the dynamic ACO conditions is relatively simple. As $\aFun^{ep}$ is strictly contracting on orbits \& its fixed points, it possesses a fixed point $x^*$. As all distances are bounded above by some value, which we will call $k^*$, the box $\boxSet^{ep}_i(k)$ can then  be defined as the set of the states which are at a distance of no more than $k^* - k$ from $x^*_i$. This is now fleshed out in more detail.

\begin{theorem}
\label{thm:ultra-implies-aco}
If $\aFun$ is a dynamic AMCO then $\aFun$ is a dynamic ACO.
\end{theorem}

\begin{proof}
Consider an epoch $e$ and set of participants $p$ such that ${(e , p) \in \configSet}$. First we prove that $\aFun^{ep}$ has a fixed point. We start by constructing the chain:
\begin{equation*}
\bot,\ \aFun^{ep}(\bot),\ (\aFun^{ep})^2(\bot),\ (\aFun^{ep})^3(\bot),\ ... 
\end{equation*}
By assumption \ref{ass:dynamic-ultra-not-participating} and the fact $\bot$ is trivially accordant, then we have that every element in the chain is in $\accSet$. Similarly by assumptions~\ref{ass:initial-set-invalid}~\&~\ref{ass:initial-set-closed} we have every element in the chain is in $\initialSet$. Therefore while $(\aFun^{ep})^k(\bot) \neq (\aFun^{ep})^{k+1}(\bot)$ then by assumption~\ref{ass:dynamic-ultra-contr-orbits} the distance between consecutive elements must strictly decrease:
\begin{equation*}
D(\bot,\aFun^{ep}(\bot)) > D(\aFun^{ep}(\bot),(\aFun^{ep})^2(\bot)) > D((\aFun^{ep})^2(\bot),(\aFun^{ep})^3(\bot)) > ...
\end{equation*}
As this is a decreasing chain in $\mathbb{N}$ it must eventually reach a $k$ such that $D(\aFun^k(\bot),\aFun^{k+1}(\bot)) = 0$. Therefore $\aFun^k(\bot) \neq \aFun^{k+1}(\bot)$ by \ref{ass:dynamic-ultra-eq} and hence $x^* = \aFun^k(\bot)$ is a fixed point and furthermore $x^* \in \initialSet$ and $x^* \in \accSet$.

By assumption~\ref{ass:dynamic-ultra-bounded} there exists an upper bound on the distance function $d_i^{ep}$ for all nodes $i$, which we will denote as $k^*$. Having established the existence of the fixed point $x^*$ and maximum distance $k^*$, we can now define $i^{th}$ component for the $k^{th}$ box as follows:
\begin{equation*}
\boxSet^{ep}_i(k) \triangleq \begin{cases}
\state_i & \text{if $k = 0$} \\
\{ \bot_i \} & \text{if $k \neq 0 \wedge i \notin p$} \\
\{ x_i \mid d_i(x_i,x^*_i) \leq \max(0, k^* - k) \} & \text{if $k \neq 0 \wedge i \in p$}
\end{cases}
\end{equation*}
We must now verify that the boxes $\boxSet^{ep}$ fulfil the conditions in Definition~\ref{def:dynamic-aco}:
\begin{enumerate}
\item \ref{ass:dynamic-aco-first-box-eq} -- $X \subseteq \boxSet^{ep}(0)$

Immediate from the first case of the definition of $\boxSet^{ep}$.

\item \ref{ass:dynamic-aco-bot-set} -- $\forall k \in \nat, i \notin p: \bot_i \in \boxSet^{ep}_{i}(k)$

Immediate from the first and second cases of the definition of $\boxSet^{ep}$.

\item \ref{ass:dynamic-aco-progress} -- $\forall k \in \nat, x \in \initialSet \cap \accSet: x \in \boxSet^{ep}(k) \Rightarrow \aFun^{ep}(x) \in \boxSet^{ep}(k+1)$

Consider some state $x \in \initialSet \cap \accSet$ and also assume that $x \in \boxSet^{ep}(k)$. We must show that for all nodes $i$ then $\aFun^{ep}_i(x) \in \boxSet^{ep}_i(k+1)$. 

If $i \notin p$ then $x_i = \bot_i$ by $x_i\in \boxSet^{ep}_{i}(k)$, and hence $\aFun^{ep}_i(x_i) = \bot_i$ by assumption~\ref{ass:dynamic-ultra-not-participating}, and so $\aFun^{ep}_i(x_i) \in \boxSet_{i}(k+1)$. Otherwise if $i \in p$ it remains to show that $d^{ep}_i(x^*_i,\aFun^{ep}_i(x)) \leq \max(0, k^* - (k + 1))$.

If $x = x^*$ then:
\begin{align*}
d^{ep}_i(x^*_i , \aFun^{ep}_i(x)) 
&= d^{ep}_i(x^*_i, \aFun^{ep}_i(x^*)) & \text{(as $x = x^*$)}\\
&= d^{ep}_i(x^*_i , x^*_i) & \text{(as $F^{ep}(x^*) = x^*$)} \\
&= 0 & \text{(by \ref{ass:dynamic-ultra-eq})} \\
&\leq \max(0, k^* - (k+1))
\end{align*}
Otherwise if $x \neq x^*$ then $d^{ep}_i(x^*_i , \aFun^{ep}_i(x)) < \max(0, k^* - k)$ as:
\begin{align*}
d^{ep}_i(x^*_i , \aFun^{ep}_i(x))
& \leq D^{ep}(x^*, \aFun^{ep}(x)) & \text{(by definition of $D$)} \\
& < D^{ep}(x^*, x) & \text{(by \ref{ass:dynamic-ultra-contr-fixed} \& \ref{ass:dynamic-ultra-not-participating})} \\
& \leq \max(0, k^* - k) & \text{(by $x \in \boxSet(k)$)}
\end{align*}
which implies that $d^{ep}_i(x^*_i , \aFun^{ep}_i(x)) \leq \max(0, k^* - (k + 1))$.

\item \ref{ass:dynamic-aco-finish} -- $\exists k^*, x^* : \forall k : k^* \leq k \Rightarrow \boxSet(k) = \{ x^*\}$

We have already established the existence of $k^*$ and $x^*$. Consider a $k \geq k^*$. If $k = 0$ then $k^* = 0$ and so the result holds trivially as $x^* \in \state$ and all points are 0 distance away from and hence equal to $x^*$ by \ref{ass:dynamic-ultra-eq} \& \ref{ass:dynamic-ultra-bounded}. Otherwise if $k \neq 0$ then as $k \geq k^*$ the definition of $\boxSet^{ep}(k)$ simplifies to:
\begin{equation*}
\boxSet^{ep}_i(k) \triangleq \begin{cases}
\{ \bot_i \} & \text{if $i \notin p$} \\
\{ x_i \mid d^{ep}_i(x_i,x^*_i) = 0 \} & \text{if $i \in p$}
\end{cases}
\end{equation*}

To prove that $x^* \in \boxSet^{ep}(k)$ we must show that for every node $i$ we have $x^*_i \in \boxSet^{ep}_i(k)$. This follows as if $i \notin p$ then as $x^*_i = \bot_i$ as $x^* \in \accSet$. Otherwise if $i \in p$ then $d^{ep}(x^*_i, x^*_i) = 0$ by \ref{ass:dynamic-ultra-eq}.

Suppose there exists another state $x \in \boxSet^{ep}(k)$. We must show that for every node $i$ we have that $x_i = x^*_i$. If $i \notin p$ then $x_i \in \boxSet^{ep}_i(k)$ implies $x_i = \bot_i$ and $\bot_i = x^*_i$ as $x^* \in \accSet$. Otherwise if $i \in p$ then $x_i \in \boxSet^{ep}_{i}(k)$ implies $d^{ep}(x_i,x^*_i) = 0$. Hence $x_i = x^*_i$ by \ref{ass:dynamic-ultra-eq}.

\end{enumerate}
Hence the conditions are satisfied and $\aFun$ is a dynamic ACO.
\end{proof}

\begin{theorem}
\label{thm:ultra-implies-convergent}
If $\aFun$ satisfies the dynamic AMCO conditions then $\async$ is convergent.
\end{theorem}
\begin{proof}
As $\aFun$ is a dynamic AMCO then $\aFun$ is a dynamic ACO by Theorem~\ref{thm:aco-implies-convergent}. Hence $\async$ is convergent by Theorem~\ref{thm:ultra-implies-aco}.
\end{proof}
\section{Formalisation in Agda}
\label{sec:agda}

Every result presented in this paper have been formalised in Agda~\cite{norell09agda}, a dependently typed language that is expressive enough that both programs and proofs may be written in it. The results cover not only the dynamic model but also include the previous static model as well. The proofs are available online~\cite{daggitt18library} as an Agda library and the library's documentation contains a guide to definitions and proofs to the corresponding Agda code.

It is hoped that the library may be of use to others in constructing formal proofs of correctness for a variety of asynchronous algorithms. The library is designed in a modular fashion so that users need not be aware of the underlying details. The library has already been used to generalise and formally verify the correctness conditions for inter-domain routing protocols with complex conditional policy languages found in~\cite{daggitt18convergence}.
\section{Conclusion}
\label{sec:conclusion}

This paper has successfully constructed a more general model for dynamic asynchronous iterations in which both the computation and the set of participants may change over time. It has generalised the ACO and AMCO conditions for the existing static model and shown that the generalisations are sufficient to guarantee the correctness of the dynamic model.

Although we have not directly explored any uses of these results in this paper, we refer interested readers to~\cite{daggitt2021formally} which contains an in-depth case study of how they may be applied to prove new theoretical results about the Bellman-Ford family of routing protocols.

There are still several open questions in regards to the theory of asynchronous iterations. For example, even in the static model questions remain about what are necessary conditions for $\delta$ to converge. \UD{}~\cite{uresin90parallel} showed that when $\state$ is finite then the ACO conditions are both necessary and sufficient for convergence. As far as the authors are aware there exist no such corresponding conditions for the case when $\state$ is infinite. 

Another obvious question is whether the dynamic ACO conditions are also necessary for the convergence of the dynamic model when $\state$ is finite. \UD{}'s static proof is essentially combinatorial in nature, building the ACO boxes $\boxSet$ such that they contain all possible states that can result from static schedules. The challenges to adapting this to the dynamic model are twofold: firstly the additional combinatorial explosion of possible states introduced by the epochs, and secondly the absence in the definition of a dynamic schedule of \UD{}'s assumption that the schedules must contain an infinite number of pseudoperiods.

\section*{Acknowledgements}

\noindent Matthew L. Daggitt was supported by an EPSRC Doctoral Training grant.

\bibliography{paper}
\end{document}